%% file: triangles.tex
\pdfoutput=1
\documentclass[letterpaper,11pt]{article}
\usepackage{fullpage}
\usepackage{amsthm}
\usepackage{amsfonts}
\usepackage{amssymb}
\usepackage{amsmath}
\usepackage{float}
\usepackage[linesnumbered,algoruled,noend]{algorithm2e}
\usepackage{paralist}
\usepackage[usenames,dvipsnames]{xcolor}

\floatstyle{ruled}
\newfloat{algorithm}{tbp}{loa}
\floatname{algorithm}{Algorithm}

\usepackage{amsthm}
\newtheorem{theorem}{Theorem}[section]
\newtheorem{lemma}[theorem]{Lemma}
\newtheorem{definition}[theorem]{Definition}
\newtheorem{corollary}[theorem]{Corollary}
\newtheorem{claim}[theorem]{Claim}
\newtheorem{remark}[theorem]{Remark}

\newcommand{\namedref}[2]{\hyperref[#2]{#1~\ref*{#2}}}
\renewcommand{\namedref}[2]{#1~\ref*{#2}}
\newcommand{\sectionref}[1]{\namedref{Section}{#1}}

\newcommand{\theoremref}[1]{\namedref{Theorem}{#1}}

\newcommand{\claimref}[1]{\namedref{Claim}{#1}}
\newcommand{\lemmaref}[1]{\namedref{Lemma}{#1}}

\newcommand{\corollaryref}[1]{\namedref{Corollary}{#1}}

\newcommand{\algref}[1]{\namedref{Algorithm}{#1}}

\newcommand{\equalityref}[1]{\hyperref[#1]{Equality~\eqref{#1}}}
\newcommand{\inequalityref}[1]{\hyperref[#1]{Inequality~\eqref{#1}}}

\newcommand{\BO}{\mathcal{O}}

\hypersetup{colorlinks,linkcolor=blue,filecolor=blue,citecolor=blue,urlcolor=blue,pdfstartview=FitH}

\begin{document}
\title{``Tri, Tri again'': Finding Triangles and Small
Subgraphs in a Distributed Setting}
\author
{Danny Dolev \\ The Hebrew University \\ Jerusalem, Israel \\
dolev@cs.huji.ac.il \and
Christoph Lenzen \\ Department for Computer Science and Applied Mathematics\\
Weizmann Institute of Science, Israel\\
clenzen@cs.huji.ac.il \and
Shir Peled \\ The Hebrew University \\ Jerusalem, Israel \\
shir.peled@cs.huji.ac.il
}
\date{}
\maketitle
\begin{quotation}
\emph{`Tis a lesson you should heed:}

\emph{Try, try, try again.}

\emph{If at first you don't succeed,}

\emph{Try, try, try again.}

\emph{(William Edward Hickson, 19th century educational writer)}\end{quotation}

\input{abs}
\input{intro}
\input{model}
\input{deterministic_general}

\input{deterministic_sparse}

\input{randomized}

\input{ack}

\bibliographystyle{abbrv}
\bibliography{../triangles}

\end{document}

%% file: abs.tex
\begin{abstract}
Let $G=(V,E)$ be an $n$-vertex graph and $M_d$ a $d$-vertex graph,
for some constant $d$. Is $M_d$ a subgraph of $G$? We consider
this problem in a model where all $n$ processes are connected to all other
processes, and each message contains up to $\BO(\log n)$ bits. A
simple deterministic algorithm that requires
$\BO(n^{(d-2)/d}/\log n)$ communication rounds is
presented. For the special case that $M_{d}$ is a triangle, we present a
probabilistic algorithm that requires an expected $\BO(\lceil n^{1/3}/(t^
{2/3}+1)\rceil)$ rounds of communication, where $t$ is the number of
triangles in the graph, and $\BO(\min\{n^{1/3}\log^{2/3}n/(t^
{2/3}+1),n^{1/3}\})$ with high probability.

We also present deterministic algorithms specially suited for sparse graphs. In
any graph of maximum degree $\Delta$, we can test for arbitrary subgraphs of
diameter $D$ in $\BO(\lceil \Delta^{D+1}/n \rceil)$ rounds. For triangles, we
devise an algorithm featuring a round complexity of $\BO(A^2/n+\log_{2+n/A^2}
n)$, where $A$ denotes the arboricity of $G$.
\end{abstract}

%% file: intro.tex
\section{Introduction}

In distributed computing, it is common to represent a distributed system as a
graph whose nodes are computational devices (or, more generally, any kind of
agents) and whose edges indicate which pairs of devices can directly
communicate with each other. Since its infancy, the area has been arduously
studying the so-called {\sc local} model (cf.~\cite{peleg00}), where the devices try to
jointly compute some combinatorial structure, such as a maximal matching or a
node coloring, of this communication graph. In its most pure form, the local
model is concerned with one parameter only: the locality of a problem, i.e., the
number of hops up to which nodes need to learn the topology and local portions
of the input in order to compute their local parts of the output---for example
this could be whether or not an outgoing edge is in the maximal
matching or the color of the node.
%\shp{unclear sentence}\cl{I'm
%not too happy with it either. Slightly changed, but I would be glad about a
%nicer formulation.} 

Considerable efforts have been made to understand the effect of bounding
the amount of communication across each edge. In particular, the {\sc congest}
model that demands that in each time unit, at most $\BO(\log n)$ bits are
exchanged over each edge, has been studied intensively. However, to the best of
our knowledge, all known lower bounds rely on
``bottlenecks''~\cite{kothapalli06,lotker06,dasSarma2011},
i.e., small edge cuts that severely constrain the total number of bits that may
be communicated between different parts of the graph. In contrast, very little
is known about the possibilities and limitations in case the communication graph
is a clique, i.e., the communication bounds are symmetric and \emph{independent}
of the structure of the problem we need to solve.
% CONGEST = LOCAL + restricted message size!
%there are no locality or routing issues at all. 
The few existing works show that, as one can expect, such a distributed system
model is very powerful: A minimum spanning tree can be found in $\BO(\log \log
n)$ time~\cite{lotker03}, with randomization nodes can send and receive up to
$\BO(n)$ messages of size $\BO(\log n)$ in $\BO(1)$ rounds, without any initial
knowledge of which nodes hold messages for which destinations~\cite{lenzen11},
and, using the latter routine, they can sort $n^2$ keys in $\BO(1)$ rounds
(where each node holds $n$ keys and needs to learn their index in the sorted
sequence)~\cite{patt-shamir11}. In general, none of these tasks can be performed
fast in the local model, as the communication graph might have a large diameter.

In the current paper, we examine a question that appears to be hard even in a
clique if message size is constrained to be $\BO(\log n)$. Given that each node
initially knows its neighborhood in an input graph, the goal is to decide
whether this graph contains some subgraph on $d\in \BO(1)$ vertices. In the
local model, this can be trivially solved by each node learning the topology up
to a constant distance;\footnote{In the local model, one is satisfied with at
least one node detecting a respective subgraph. Requiring that the output is
known by all nodes results in the diameter being a trivial lower bound for any
meaningful problem.} in our setting, this simple strategy might result in a
running time of $\Omega(n/\log n)$, as some (or all) nodes may have to learn
about the entire graph and thus need to receive $\Omega(n^2)$ bits. We devise a
number of algorithms that achieve much better running times. These algorithms
illustrate that efficient algorithms in the contemplated model need to strive
for balancing the communication load, and we show some basic strategies to do
so. We will see as a corollary that it is possible for all nodes to learn about
the entire graph within $\BO(\lceil |E|/n \rceil)$ rounds and therefore locally
solve any (computable) problem on the graph; this refines the immediately obvious
statement that the same can be accomplished within $\Delta$ (where $\Delta$
denotes the maximum degree of $G$) rounds by each node sending its complete list
of neighbors to all other nodes. For various settings, we achieve running times
of $o(|E|/n)$ by truly distributed algorithms that do not require that (some)
nodes obtain full information on the entire input.

Apart from shedding more light on the power of the considered model, the
detection of small subgraphs, sometimes referred to as \emph{graphlets} or
\emph{network motifs}, is of interest in its own right. Recently, this topic
received growing attention due to the importance of recurring patterns in
man-made networks as well as natural ones. Certain subgraphs were found to be
associated with neurobiological networks, others with biochemical ones, and
others still with human-engineered networks~\cite{milo02}. Detecting network
motifs is an important part of understanding  biological networks, for instance,
as they play a key role in information processing mechanisms of biological
regulation networks.
Even motifs as simple as triangles are of interest to the
biological research community as they appear in gene regulation networks, where
what a graph theorist would call a directed triangle is often referred to as a
Feed-Forward Loop. In recent years, the network motifs approach to studying
networks lead to development of dedicated algorithms and software tools. Being
of highly applicative nature, algorithms used in such context are usually
researched from an experimental point of view, using naturally generated data
sets~\cite{kashtan04}.

Triangles and triangle-free graphs also play a central role in combinatorics.
For example, planar triangle-free graphs are long since known to be
3-colorable~\cite{groetzsch59}. The implications of triangle finding and
triangle-freeness motivated extensive research of algorithms, as well as lower
bounds, in the centralized model. Most of the work done on these problems falls
into one of two categories: subgraph listing and property testing. In subgraph
listing, the aim is to list all copies of a given subgraph. The number of copies
in the graph, that may be as high as $\Theta (n^3)$ for triangles,
sets an obvious lower bound for the running time of such algorithms, rendering
instances with many triangles harder in some sense~\cite{chiba85}. Property
testing algorithms, on the other hand, distinguish with some probability between
graphs that are triangle-free and graphs that are far from being triangle-free,
in the sense that a constant fraction of the edges has to be removed in order
for the graph to become triangle-free~\cite{alon02, alon08}. Although soundly
motivated by stability arguments, the notion of measuring the distance from
triangle-freeness by the minimal number of edges that need to be removed seems
less natural than counting the number of triangles in the graph. Consider for
instance the case of a graph with $n$ nodes comprised of $n-2$ triangles,
%\shp{should we add an illustration?}\cl{I don't think so. It's an extremely
%simple graph. Any qualified reader should get the idea right away.}, 
all sharing
the same edge. From the property testing point of view, this graph is very close
to being triangle free, although it contains a linear number of triangles. Some
query-based algorithms were suggested in the centralized model, where the
parameter to determine is the number of triangles in the graph.
The lower bounds for such algorithms assume restrictions on the type of
queries\footnote{For instance, in~\cite{gonen11} the query model requires that
edges are sampled uniformly at random.} that cannot be justified in our
model~\cite{gonen11}.

\textbf{Detailed Contributions.} 
In \sectionref{sec:tri0}, we start out by giving a family of deterministic
algorithms that decide whether the graph contains a $d$-vertex subgraph within
$\BO(n^{(d-2)/d})$ rounds. In fact, these algorithms find \emph{all} copies of
this subgraph and therefore could be used to count the exact number of
occurrences. They split the task among the nodes such that each node is
responsible for checking an equal number of subsets of $d$ vertices for being
the vertices of a copy of the targeted subgraph. This partition of the problem is
chosen independently of the structure of the graph. Note that even the trivial
algorithm that lets each node collect its $D$-hop neighborhood and test it for
instances of the subgraph in question does not satisfy this property. Still it
exhibits a structure that is simple enough to permit a deterministic
implementation of running time $\BO(\lceil \Delta^{D+1}/n \rceil)$, where
$\Delta$ is the maximum degree of the graph, given in \sectionref{sec:sparse}. For the special
case of triangles, we present a more intricate way of checking neighborhoods
that results in a running time of $\BO(A^2/n+\log_{2+n/A^2} n)\subseteq
\BO(|E|/n+\log n)$, where the arboricity $A$ of the graph denotes the minimal
number of forests into which the edge set can be decomposed. While always $A\leq
\Delta$, it is possible that $A\in \BO(1)$, yet $\Delta\in \Theta(n)$ (e.g.\ in
a graph that is a star). Moreover, any family of graphs excluding a fixed minor
has $A\in \BO(1)$~\cite{deo98}, demonstrating that the arboricity is a much less
restrictive parameter than $\Delta$. Note also that the running time bound in
terms of $|E|$ is considerably weaker than the one in terms of $A$; it serves to
demonstrate that in the worst case, the algorithm's running time essentially
does not deteriorate beyond the trivial $\BO(|E|/n)$ bound.

All our deterministic algorithms systematically check for subgraphs by either
considering all possible combinations of $d$ nodes or following the edges of the
graph. If there are many copies of the subgraph available, it can be much more
efficient to randomly inspect small portions of the graph. In
\sectionref{sec:randomized}, we present a triangle-finding algorithm that does
just that, yielding that for every $\varepsilon\geq 1/n$ and a graph containing
$t\geq 1$ triangles, a triangle will be found with probability at least
$1-\varepsilon$ within $\BO((n^{1/3}\log^{2/3}\varepsilon^{-1})/t^{2/3}+\log n)$
rounds; we show this analysis to be tight.

All our algorithms are uniform, i.e., they require no prior knowledge of
parameters such as $t$ or $A$. Interleaving them will result in an asymptotic
running time that is  bounded by the minimum of all the individual results. All
proofs are omitted from this extended abstract due to lack of space, and are
detailed in full in the appendix.

%% file: model.tex
\section{Model and Problem}

Our model separates the computational problem from the communication model. The
set $V=\{1,\ldots,n\}$ represents the nodes of a distributed system. With
respect to communication, we adhere to the synchronous {\sc congest} model as
described in \cite{peleg00} on the complete graph on the node set $V$, i.e., in
each computational round, each node may send (potentially different) $\BO(\log
n)$ bits to each other node. We do not consider the amount of computation
performed by each node, however, for all our algorithms it will be polynomially
bounded. Instead, we measure complexity in the number of rounds until an
algorithm terminates.\footnote{Note that it is trivial to make all nodes
terminate in the same round due to the full connectivity.} Let $G=(V,E)$ be an
arbitrary graph on the same vertex set, representing the computational problem
at hand. Initially, every node $i\in V$ has the list ${\cal N}_i:=\{j\in
V\,|\,\{i,j\}\in E\}$ of its neighbors in $G$, but no further knowledge of $G$.

The computational problem we are going to consider throughout this paper is the
following: Given a graph $M_{d}$ on $d\in \BO(1)$ vertices, we wish to
discover whether $M_{d}$ is a subgraph of $G$.

%% file: deterministic_general.tex
\section{Deterministic Algorithms for General Graphs}\label{sec:tri0}

During our exposition, we will discuss the issues of what to communicate and how
to communicate it separately. That is, given
sets of $\BO(\log n)$-sized messages at all nodes satisfying certain properties,
we provide subroutines that deliver all messages quickly, and use these
subroutines in our algorithms. We start out by giving a very efficient
deterministic scheme provided that origins and destinations of all messages are
initially known to \emph{all} nodes. We then will show that this scheme can be
utilized to find all triangles or other constant-sized subgraphs in sublinear
time.

\subsection{Full-Knowledge Message Passing}

For a certain limited family of algorithms that we call \emph{oblivious
algorithms}, it is possible to exploit the full capacity of the communication
system, i.e., provided that no node sends or receives more than
$n$ messages, all messages can be delivered in two rounds.

\begin{definition}
A distributed algorithm $\cal A$ in our model is said to be \emph{oblivious}
if the sources and destinations of all messages are determined in
advance, regardless of the input graph $G$, and each source can determine the
content of its messages from its input.
\end{definition}

Take for example an algorithm in which every node $i$ sends every node $j$ a bit
string stating for each node $k$ whether it is a neighbor of $i$. This algorithm
is clearly oblivious and results in all nodes having complete knowledge of
the structure of $G$.

As communication is peer-to-peer, sending messages to different nodes can be
executed in parallel. If all nodes execute the above suggested routine, after
$n$ rounds every node gets all lists of immediate neighbors of nodes, and can
therefore reconstruct the graph locally. We will see later on,
in~\sectionref{sec:sparse}, that a similar algorithm can be realized more
efficiently using a more evolved communication strategy.

We now turn to describing our communication pattern for oblivious algorithms.
To this end, we will need the following claim that is a corollary of Hall's
marriage theorem.
\begin{claim}
\label{cor:(Hall's-theorem)-Every}Every $d$-regular bipartite multigraph
is a disjoint union of $d$ perfect matchings.\end{claim}
\begin{proof}
By induction on $d$. For $d=1$ the graph is a perfect matching by
definition.

Assume that the claim holds for some $d$, and let $H=(L,R,E)$ be a
$(d+1)$-regular bipartite graph. Let $S\subseteq L$ be some set of vertices, and
define $\Gamma(S):=\{ u\in R:\exists v\in S\text{ s.t.~\ensuremath{(v,u)\in
E}}\}$. By regularity, the sum of degrees in $S$ is exactly $(d+1)|S|$, and by
the pigeonhole principle and regularity $|\Gamma(S)|\geq(d+1)|S|/(d+1)=|S|$,
satisfying Hall's marriage condition thus implying that a perfect matching
exists. Removing the perfect matching found from the graph leaves a $d$-regular
bipartite graph that is a disjoint union of $d$ perfect matchings by the
induction hypothesis. Adding those $d$ perfect matchings to the one just
obtained completes the proof.
\end{proof}
\begin{lemma}\label{lemma:2_round}
Given a bulk of messages, such that:
\begin{compactenum}
\item The source and destination of each message is known in advance to all
nodes, and each source knows the contents of the messages to sent.
\item No node is the source of more than $n$ messages.
\item No node is the destination of more than $n$ messages.
\end{compactenum}
A routing scheme to deliver all messages within 2 rounds can be found
efficiently.
\end{lemma}
\begin{proof}
WLOG we assume every node is the source of exactly $n$ messages,
and it is the destination of exactly $n$ messages as well (having a
node ``sending message to itself'' is not a problem). We will label every
message to node $i$ with a different $j\in \{1,...,n\}$ and denote the messages
to node $i$ according to this labeling by $m_{i,1},m_{i,2},...,m_{i,n}$. 

We define a \emph{good} labeling to be such that no node initially
holds two messages labeled $m_{j,k}$ and $m_{l,k}$ for some $l,j,k$ with $l\neq
j$. Assuming we start with a good labeling, we argue that the message
passing algorithm whose pseud-code is given in \algref{algo:2_round} terminates
successfully after two rounds. We will later show that a good labeling is always attainable.

\begin{algorithm}[t!]
\caption{Deterministic Message Passing at node $i$ holding message set $S$.}
\label{algo:2_round}
$S':=\emptyset$, $S'':=\emptyset$\\
// first stage (distribution)\\
\For{$m_{j,k}\in S$}{
  send $m_{j,k}$ to node $k$
}
\For{received message $m$}{
  $S':=S'\cup \{m\}$
}
// second stage (delivery)\\
\For{$m_{j,k}\in S'$}{
  send $m_{j,k}$ to node $j$
}
\For{received message $m$}{
  $S'':=S''\cup \{m\}$
}
return $S''$
\end{algorithm}

If our labeling is indeed good, then during the first stage every
node sends at most a single message to each of the other nodes, and
therefore can dispose of all the messages in $S$ within the first
round. Due to the unique labeling of the messages, after the first stage
node $i$ holds all messages of type $m_{k,i}$, and since
there is at most one such message for each $k$, all of them are emitted
within a single round in the second stage. Clearly, the labeling also ensures
that the returned set $S''$ will contain exactly the messages whose destination
is $i$.

It remains to show that we can find a good labeling. Recall that sources
and destinations are known in advance to all nodes, so each node can compute the
labeling locally. If all nodes use the same deterministic algorithm, this will
result in all nodes using the exact same labeling.

Let $B=(L,R,E)$ be a bipartite multigraph, where $|L|=|R|=n$. We denote $L=\{
l_{1},...,l_{n}\} $ and $R=\{ r_{1},...,r_{n}\} $. For every message in the
initial bulk with source $i$ and destination $k$ we add an edge $(l_{i},r_{k})$
to $E$. $B$ is clearly an $n$-regular multigraph, and by Claim
\ref{cor:(Hall's-theorem)-Every} it is a disjoint union of $n$ perfect
matchings. We now choose a perfect matching in this graph, remove its edges and
label the messages represented by those edges thus: for every edge
$(l_{i},r_{k})$ in the matching we label its corresponding message $m_{k,1}$.
After removing those edges we find another perfect matching, and for every edge
$(l_{i},r_{k})$ in it we label the corresponding message $m_{k,2}$ and so on,
until we remove the $n$th perfect matching from the graph. Since a perfect
matching is easy to find (using maximal-flow algorithms), a good labeling can be
found efficiently.\end{proof}
\begin{corollary}
\label{cor:An-oblivious-algorithm takes T(n)/n}An oblivious algorithm
in which each node sends and receives at most $T(n)$ messages
can be completed within $2\lceil T(n)/n\rceil$ rounds, by repeatedly using the
message passing routine described above.
\end{corollary}

\subsection{TriPartition - Finding triangles deterministically}

Next, we present an algorithm that finds whether there are triangles in
$G$. The algorithm is not oblivious, since in the final step every
node broadcasts whether it found a triangle or not to all other nodes.
This last broadcast message is obviously dependent on other messages
transferred throughout the algorithm, therefore it violates the obliviousness
requirement that the order of the messages will not matter. However,
having every node broadcast its results takes a single round only. The first
part of the algorithm is oblivious, allowing us to apply the message passing
algorithm previously stated to it. As the oblivious part of the algorithm
terminates, we run the final broadcasting round.

Let $S\subseteq2^{V}$ be a partition of $V$ into equally sized subsets of
cardinality $n^{2/3}$. We write $S=\{S_1,...,S_{n^{1/3}}\} $. To each node $i\in
V$ we assign a distinct triplet from $S$ denoted $S_{i,1},S_{i,2},S_{i,3}$
(where repetitions are admitted). Clearly, for any subset of three nodes there
is a triplet such that each node is element of one of the subsets in the
triplet, showing the following claim.

\begin{claim}\label{claim:TriPartition correctness}
For each triangle $\{t_1,t_2,t_3\}$ in $G$, there is some node $i$ such that
$t_1\in S_{i,1}$, $t_2\in S_{i,2}$, and $t_3\in S_{i,3}$.
\end{claim}
\begin{proof}
Each node checks for triangles that are contained in its triplet of subsets by
executing \emph{TriPartition}, whose pseudo-code is given in
\algref{algo:TriPartition}.
\end{proof}

\begin{algorithm}[t!]
\caption{TriPartition at node $i$.}
\label{algo:TriPartition}
$E_i:=\emptyset$\\
\For{$1\leq j<k\leq3$}{
  \For{$l\in S_{i,j}$}{
    retrieve ${\cal N}_l\cap S_{i,k}$\\
    \For{$m\in {\cal N}_l\cap S_{i,k}$}{
      $E_i:= E_i\cup\{l,m\}$
    }
  }
}
\If{there exists a triangle in $G_i:=(V,E_i)$}{
  send ``triangle'' to all nodes
}
\If{received ``triangle'' from some node}{
  return \textbf{true}
}
\Else{
  return \textbf{false}
}
\end{algorithm}

\begin{theorem}
\label{thm:TriPartition-determines-correctly in n^1/3}\emph{TriPartition} determines
correctly whether there exists a triangle in $G$ and can be implemented within
$\BO(n^{1/3})$ rounds.
\end{theorem}
\begin{proof}
Correctness follows from \claimref{claim:TriPartition correctness}, as node $i$
collects exactly the edges between pairs of subsets in its triplet. The round
complexity is deduced as follows. Since the assignment of set triplets is
static, each node $i$ knows which nodes need to learn about which of its
neighbors. Since there are $n^{1/3}$ subsets of size $n^{2/3}$, each of which
participates in $n^{1/3}$ triplets involving the subset containing $i$, the node
needs to transmit at most $n^{4/3}$ messages.\footnote{Clearly, a neighbor can
be encoded using $\log n$ bits.} On the other hand, each node needs to learn
about less than $\binom{3}{2} n^{4/3}$ edges, one for each pair of nodes from
two of its subsets. By \corollaryref{cor:An-oblivious-algorithm takes T(n)/n},
this information can thus be communicated within $\BO(n^{1/3})$ rounds. The
algorithm terminates one additional round later, completing the proof.
\end{proof}
\begin{remark}
\label{rem:partial obliviousness results in log factor}The fact that (except
for the potential final broadcast) the entire communication pattern of
\emph{TriPartition} is predefined enables to refrain from including any node
identifiers into the messages. That is, instead of encoding the respective
sublist of neighbors by listing their identifiers, nodes just send a $0-1$ array
of bits indicating whether a node from the respective set from $S$ is or is not
a neighbor in $G$. The receiving node can decode the message because it is
already known in advance which bit stands for which pair of nodes. We may hence
improve the round complexity of \emph{TriPartition} to $\BO(n^{1/3}/\log n)$.
\end{remark}

\subsection{Generalization for \texorpdfstring{$d$}{d}-cliques}

\emph{TriPartition} generalizes easily to an algorithm we call \emph{dClique0}
that finds $d$-cliques (as well as any other subgraph on d-vertices).
We choose $S$ to be a partition of $V$ into equal size subsets of
cardinality $n^{(d-1)/{d}}$, resulting in $S=\{S_1,...,S_{n^{1/d}}\} $.
Each node now examines the edges between all pairs of some $d$-sized
multisubset of $S$ (as we did for $d=3$ in \emph{TriPartition}). Since there
are exactly $|S|^{d}=n$ such multisets, all possible $d$-cliques
are examined. Every node needs to receive the list of edges for all
$\binom{d}{2}$ pairs, each containing at most
$(n^{(d-1)/d})^2$ edges, thus every node needs to send
and receive at most $\BO(n^{(2d-2)/d})$ messages.
\begin{theorem}
dClique0 determines correctly whether there exists a d-clique (or any given
d-vertex graph) in $G$ within $\BO(n^{(d-2)/d}/\log n)$ rounds.
\end{theorem}
\begin{proof}
Similarly to the 3-vertex case, we apply Corollary \ref{cor:An-oblivious-algorithm takes T(n)/n},
and due to obliviousness, we may assume all messages are sequences
of bits as in Remark \ref{rem:partial obliviousness results in log factor}.
\end{proof}

%% file: deterministic_sparse.tex
\section{Finding triangles in sparse graphs}\label{sec:sparse}

In graphs that have $o(n^2)$ edges, one might hope to obtain faster algorithms.
However, the algorithms from the previous section have congestion at the node
level, i.e., even if there are few edges in total, some nodes may still have to
send or receive lots of messages. Hence, we need different strategies for sparse
graphs. In this section, we derive bounds depending on parameters that reflect
the sparsity of graphs.

\subsection{Bounded Degree}

We start with a simple value, the \emph{maximum degree} $\Delta:=\max_{i\in
V}\delta_i$, where the \emph{degree of node $i$} $\delta_i:=|{\cal N}_i|$. If
$\Delta$ is relatively small, then every node can simply exchange its neighbors
list with all its neighbors. We refer to this as \emph{TriNeighbors} algorithm,
whose pseudo-code is given in \algref{algo:TriNeighbors}. In a graph with
bounded $\Delta$, it will be much faster than \emph{dClique0} algorithm.

\begin{algorithm}[t!]
\caption{TriNeighbors at node $i$.}
\label{algo:TriNeighbors}
$E_i:=\emptyset$\\
\For{$j\in V$ s.t.~$(i,j)\in E$}{
  retrieve ${\cal N}_j$\\
  \For{$k\in {\cal N}_j$}{
    $E_i:= E_i\cup\{j,k\}$
  }
}
\If{there exists a triangle in $G_i:=(V,E_i)$}{
  send ``triangle'' to all nodes
}
\If{received ``triangle'' from some node}{
  return \textbf{true}
}
\Else{
  return \textbf{false}
}
\end{algorithm}

We use an elegant message-passing technique, suggested to us by Shiri
Chechik~\cite{chechik11}. Assuming that (i) no node is the source of more than
$n$ messages in total, (ii) no node is the destination of more than $n$
messages, and (iii) every node sends the exact same messages to all of the
destinations for its messages, it delivers all messages in $3$ rounds. This is
done by first having each node distribute its messages evenly, in a Round-Robin
fashion, to all other nodes in the graph. In the second phase, messages are
retrieved in a similar Round-Robin process. This divides the communication load
evenly, resulting in an optimal round complexity. Assuming that for each node
$i$ we have the set of its $k(i)$ messages $M_i=\{m_{i,1},\ldots,m_{i,k(i)}\}$,
let $D_i$ denote its recipients list. With these notations, Chechik's
\emph{Round-Robin Messaging} algorithm is given in \algref{algo:round_robin}.

\begin{algorithm}[t!]
\caption{Round-Robin-Messaging at node $i$.}\label{algo:round_robin}
$R:=\emptyset$ // collects output\\
$S:=\emptyset$ // collects source nodes and \#messages for $i$\\
\For{$j\in V$}{
  send $m_{i,j\operatorname{mod}k(i)}$ to $j$\\
  \If{$j\in D_i$}{
    send ``notify $k(i)$'' to $j$
  }
}
\For{``notify $k(j)$'' received from $j$}{
  $S:=S\cup (j,k(j))$
}
$l:=1$\\
\For{$(j,k(j))\in S$}{
  \For{$k\in \{1,\ldots,k(j)\}$}{
    send ``request message from $j$'' to $l$\\
    $l:=l+1$
  }
}
\For{received ``request message from $j$''}{
  send $m_{j,i\operatorname{mod}k(j)}$ to $j$
}
\For{received message $m$}{
  $R:=R\cup\{m\}$
}
return $R$
\end{algorithm}

\begin{lemma}
\label{lem:Shiri's algorithm}
Given a bulk of messages in which:
\begin{compactenum}
  \item Every node is the source of at most $n$ messages.
  \item Every node is the destination of at most $n$ messages.
  \item Every source node sends exactly the same information to all of its
  destination nodes and knows the content of its messages.
\end{compactenum}
\emph{Round-Robin-Messaging} delivers all messages in $3$ rounds.
\end{lemma}
\begin{proof}
In the first loop of the algorithm every node sends one message to every other
node; note that it is feasible to send both the message
$m_{i,j\operatorname{mod}k(i)}$ and a potential notification at the same time.
The cyclic nature of the message distribution in this first loop assures that
any consecutive $k(i)$ nodes together hold all $k(i)$ messages of node $i$,
exactly one at each node. By Condition~1, $k(i)\leq n$ for each node $i$,
i.e., each node indeed sends out all its messages. By Condition~2,
for each node the querying loop will request at most one message from each node.
Since exactly $k(j)$ messages are requested from a node~$j$, the set of messages
retrieved in the second last loop contains $M_j$. By Condition~3 and due to the
previous notification of destination nodes, this is exactly the set of messages
to be received from $j$. This shows correctness of the algorithm. As we also
argued that in total three communication rounds are required, this shows the
statement of the lemma.
\end{proof}
\begin{corollary}
Using \emph{Round-Robin-Messaging}, the complete structure of the graph can be
known to all nodes in $\BO{\lceil |E|/n \rceil}$ rounds.
\end{corollary}

Algorithm \emph{TriNeighbors} satisfies all the conditions of
\lemmaref{lem:Shiri's algorithm}. We conclude that, employing
\emph{Round-Robin-Messaging}, the round complexity of \emph{TriNeighbors} becomes
$\BO(\lceil \Delta^2/n \rceil)$. If $\Delta\in \BO(\sqrt{n})$ then the round
complexity is $\BO(1)$, and clearly optimal. More generally, any subgraph of
diameter\footnote{The diameter of the graph is the maximum shortest path length
over all pairs of nodes.} $D\in \BO(1)$ can be detected by each node exploring
its $D$-hop neighborhood.
\begin{corollary}
We can test for subgraphs of diameter $D$ in $\BO(\lceil \Delta^{D+1}/n \rceil)$
rounds.
\end{corollary}

\subsection{Bounded Arboricity}

The arboricity $A$ of $G$ is defined to be the minimum number of forests
on $V$ such that their union is $G$. Note that always $A\leq \Delta$,
and for many graphs $A\ll \Delta$. The arboricity bounds the number of
edges in any subgraph of $G$ in terms of its nodes.
% \begin{lemma}
% \label{lem:Arboricity degree}Let $G=(V,E)$ be some graph
% with arboricity $A$, then $\left|\left\{ v\in V:deg(v)<4A\right\}\right|>\left|V\right|/2$\end{lemma}
% \begin{proof}
% Assume by way of negation that $\left|\left\{ v\in
% V:deg(v)<4A\right\} \right|\leq\left|V\right|/2$,  it follows that
% $\sum_{v\in V}deg(v)\geq\left|V\right|/2\cdot4A=2\left|V\right|A$.
% However, a graph with arboricity $A$, being a union of $A$ forests,
% has at most $(\left|V\right|-1)A$ edges, leading to a
% contradiction.\end{proof}
% \begin{remark}
% \label{fac:Arboricity of subgraphs}The arboricity of a graph bounds
% from above the arboricity of all its subgraphs. 
% \end{remark}
We exploit this property to devise an arboricity-based algorithm for
triangle finding that we call \emph{TriArbor}.

\subsubsection{An overview of the TriArbor algorithm}

We wish to employ the same strategy used by the naive $TriNeighbors$, that is
``asking neighbors for their neighbors'', in a more careful manner, so as to
avoid having high degree nodes send their entire neighbor list to many nodes.
This is achieved by having all nodes with degree at most $4A$ send their
neighbor list to their neighbors and then shut down. In the next iteration, the
nodes that have a degree at most $4A$ in the graph induced by the still active
nodes do the same and shut down. As $2An'$ uniformly bounds the sum of degrees
of any subgraph of $G$ containing $n'$ nodes, in each iteration at least half of
the remaining nodes is shut down. Hence, the algorithm will terminate within
$\BO(\log n)$ iterations. In order to control the number of messages sent in
each iteration, we consider triangles involving at least one node of low degree
(in the induced subgraph of the still active nodes). As we will find a triangle
once any of its nodes' degrees becomes smaller than $4A$, all triangles are
will be detected.

Obviously, no node of low degree will have to send more than $4A$ messages in
this scheme. However, it may be the case that a node receives more than $4A$
messages in case it has many low-degree neighbors. To remedy that, low-degree
nodes avoid sending their neighbor list to their high-degree neighbors directly,
and instead send them to intermediate nodes we call \emph{delegates}. The
delegates share the load of testing their associated high-degree node's
neighborhood for triangles involving a low-degree node.

Note that in the presented form, the algorithm is not uniform, i.e., it is
assumed that $A$ is known. We will later discuss how to remove this assumption
and slightly improve its round complexity at the same time.

\subsubsection{TriArbor algorithm}

\subsubsection*{Choosing delegates}

In each iteration, every delegate node will be assigned to a unique high-degree
node, i.e., a node of degree larger than $4A$ in the subgraph induced by the
nodes that are still active. In the following, we will discuss a single
iteration of the algorithm. Denote by $G':=(V',E')$ some subgraph of $G$ on $n'$
nodes, where WLOG $V'=\{1,\ldots,n'\}$. Define $\delta'_i$, $\Delta'$,
${\cal N}'_i$, etc.\ analogously to the respective values without a prime, but
with respect to $G'$ instead of $G$. We would like to assign to each node $i$
exactly $\lceil\delta_i'/(4A)\rceil$ delegates such that each delegate is
responsible for up to $4A$ of the respective high-degree node's neighbors.
\begin{claim}\label{claim:enough_labels}
At least $n'/2$ of the nodes have degree at most $4A$ and the number of assigned
delegates is bounded by $n'$.
\end{claim}
\begin{proof}
We have that
\begin{equation*}
|\{i\in V'\,|\,\delta'_i>4A\}|\leq \frac{1}{4A}\sum_{i\in V'}\delta'_i
\leq \frac{|E'|}{2A}< \frac{n'}{2}.
\end{equation*}

Therefore,
\begin{equation*}
\sum_{\substack{i\in V'\\\delta'_i>4A}}
\left\lceil\frac{\delta'_i}{4A}\right\rceil \leq
\frac{n'}{2}+\frac{1}{4A}\sum_{i=1}^{n'}\delta'_i \leq
\frac{n'}{2}+\frac{|E'|}{2A} < n',
\end{equation*}
i.e., less than $n'$ delegates are required.
\end{proof}

Moreover, the assignment of delegates to high-degree nodes can be computed
locally using a predetermined function of the degrees $\delta'_i$. Thus, if
every node communicates its degree $\delta'_i$, all nodes can determine locally
the assignment of delegates to high-degree nodes in a consistent manner.

\subsubsection*{The algorithm}

\algref{algo:tri_arbor_iteration} shows the pseudocode of one iteration of
\emph{TriArbor}. The complete algorithm iterates until for all nodes
$\delta_i'=0$ and outputs ``true'' if in one of the iterations a triangle was
detected and ``false'' otherwise.

\begin{algorithm}[t!]
\caption{One iteration of TriArbor at node $i$.}\label{algo:tri_arbor_iteration}
// compute delegates\\
send $\delta'_i$ to all other nodes\\
compute assignment of delegates to high-degree nodes and neighbor sublists\\
// high-degree nodes distribute their neighborhood\\
\If{$\delta'_i>4A$}{
  partition ${\cal N}'_i$ into $\lceil \delta'_i/4A\rceil$ lists of length at
  most $4A$\\
  send each sublist to the computed delegate\\
  \For{$j\in {\cal N}'_i$}{
    notify $j$ of the delegate assigned to it // only $i$ knows the order of
    ${\cal N}'_i$, hence communication required
  }
}
// let all delegates learn about ${\cal N}'_j$\\
\If{$i$ is delegate of some node $j$}{
  denote by $D_j$ the set of delegates of $j$\\
  denote by $L_{j,i}\subset {\cal N}'_j$ the sublist of neighbors received from
  $j$\\
  \For{$k\in D_j$}{
    send $L_{j,i}$ to $k$
  }
  \For{received sublist $L_{j,k}$}{
    ${\cal N}'_j:={\cal N}'_j\cup L_{j,k}$
  }
}
// low-degree nodes distribute their neighborhoods\\
\If{$\delta'_i\leq 4A$}{
  \For{$j\in {\cal N}'_i$}{
    \If{$\delta'_j\leq 4A$}{
      send ${\cal N}'_i$ to $j$ // low-degree nodes can handle load themselves
    }
    \Else{
      send ${\cal N}'_i$ to the delegate of $j$ assigned to $i$
    }
  }
}
// check for triangles\\
\For{received ${\cal N}'_j$ (from $j$ with $\delta'_j\leq 4A$)}{
  \If{${\cal N}'_i\cap {\cal N}'_j\neq \emptyset$}{
    send ``triangle found'' to all nodes // detected triangles involving two
    low-degree nodes
  }
  \ElseIf{$i$ is delegate of $k$ and ${\cal N}'_j\cap {\cal N}'_k\neq
  \emptyset$}{
    send ``triangle found'' to all nodes // detected triangle involving one
    low-degree node
  }
}
\If{received ``triangle found''}{
  return \textbf{true}
}
\Else{
  return \textbf{false}
}
\end{algorithm}

\begin{claim}\label{claim:tri_arbor_termination}
\emph{TriArbor} terminates within $\lceil\log n\rceil$ iterations.
\end{claim}
\begin{proof}
Follows directly from \claimref{claim:enough_labels}, as in each iteration at
least half of the nodes are eliminated.
\end{proof}

\begin{lemma}\label{lem:tri_arbor_correct}
\emph{TriArbor} correctly decides whether the graph contains a
triangle or not.
\end{lemma}
\begin{proof}
Clearly, there are no false positives, as in each iteration, nodes will only
claim that a triangle is found if they learned about an edge connecting two
nodes in the same neighborhood (either their own or the node whose delegate they
are).

Recall that by \claimref{claim:enough_labels}, there are sufficiently many
delegates available, and we observed that the assignment can be computed as the
same function of the (current) degrees.

Now, assume the graph contains some triangle $\{i_1,i_2,i_3\}$. There must be
some iteration in which one of the nodes, say $i_1$, has degree
$\delta'_{i_1}\leq 4A$ and the triangle is still in the subgraph induced by
active nodes: By \claimref{claim:tri_arbor_termination}, eventually all nodes
get eliminated, while each edge connecting two high-degree nodes will still be
present in the subgraph induced by the active nodes of the next iteration.

We distinguish two cases. If in the respective iteration it also holds that
$\delta'_{i_2}\leq 4A$, then $i_1$ will send $i_2$ its neighbor list (with
respect to the induced subgraph), and $i_2$ will detect the triangle.
Otherwise, $i_1$ will send its current neighbor list to one of $i_2$'s
delegates. As $i_2$ splits its neighbor list and distributes it among its
delegates, which share their sublist with all other delegates, this delegate
will detect the triangle. Hence, in both cases, the triangle will eventually be
discovered, this information be spread among the nodes, and all nodes will
compute the correct output.
\end{proof}

\subsubsection*{Round complexity of TriArbor}

We examine the time complexity of one iteration of the algorithm. Obviously,
announcing degrees takes a single round only.
\begin{claim}\label{cl:distribute_high_degree}
The distribution of high-degree nodes' neighborhoods can be performed in two
rounds.
\end{claim}
\begin{proof}
Every node $i$ with $\delta'_i>4A$ partitions its neighbor list and sends it,
totalling in at most $\delta'_i<n'$ messages. As each node is delegate of at
most one node, no more than $4A$ messages need to be received. Observe that
since all nodes are aware of the assignment of delegates as well as all node
degrees, we can apply \lemmaref{lemma:2_round} to see that all messages can be
delivered in two rounds. Notifying neighbors of their assigned delegates takes
one message for each neighbor. However, both tasks are independent, therefore we
can merge the respective messages, resulting in a total of two rounds.
\end{proof}
\begin{claim}\label{cl:delegates_exchange}
Exchanging neighborhood sublists between delegates can be implemented in four
rounds.
\end{claim}
\begin{proof}
Every delegate holds a sublist of at most $4A$ of the neighbors of the node $i$
it has been assigned to. Hence, it needs to send at
most $\lceil\delta'_i/4A\rceil 4A<2\delta'_i<2n'$ messages. Similarly, it
receives less than $2n'$ messages. As delegates are aware of the number of
messages to exchange, \lemmaref{lemma:2_round} shows that we can implement this
communication in four rounds.
\end{proof}

\begin{claim}\label{cl:distribute_low_degree}
The distribution of low-degree nodes' neighborhoods can be performed in
$3\lceil 32 A^2/n\rceil$ rounds.
\end{claim}
\begin{proof}
Every node $i$ with $\delta'_i\leq 4A$ sends $\delta'_i$ messages
to each of its low-degree neighbors and to one delegate of each high-degree
neighbor, i.e., at most $16A^2$ messages. Similarly, both low-degree nodes and
delegates receive at most $16A^2$ messages. As the low-degree nodes send their
entire neighborhood to all destinations, applying \lemmaref{lem:Shiri's
algorithm} repeatedly yields that this communication can be performed in
$3\lceil 32 A^2/n\rceil$ rounds (note that nodes may have to receive $32 A^2/n$
messages because they may have low degree and be delegate at the same time).
\end{proof}
Finally, announcing a found triangle takes one more round. All in all, we get
the following result.
\begin{theorem}\label{theorem:tri_arbor}
Algorithm~\emph{TriArbor} is correct. Using our \emph{Deterministic Message
Passing} and \emph{Round-Robin-Messaging} algorithms, it can be implemented with
a running time of $\BO(\lceil A^2/n\rceil \log n)$ rounds.
\end{theorem}
\begin{proof}
Correctness was shown in \lemmaref{lem:tri_arbor_correct}. Combining
Claims~\ref{cl:distribute_high_degree}, \ref{cl:delegates_exchange},
and~\ref{cl:distribute_low_degree}, we see that a single iteration of the
algorithm can be implemented with running time $\BO(\lceil A^2/n \rceil)$. By
\claimref{claim:tri_arbor_termination}, the total running time is thus bounded
by $\BO(\lceil A^2/n\rceil \log n)$ rounds.
\end{proof}
\begin{corollary}\label{cor:parallelizing-TriArbor}
The iterations of TriArbor can be parallelized, reducing the round complexity to
$\BO(A^2/n + \log n)$.
\end{corollary}
\begin{proof}
We first let all nodes execute the a short announcement phase, whose
pseudo-code is given in \algref{algo:degree announcement}, storing all received
values.
\begin{algorithm}[t!]
\caption{QuickDecomposition at node $i$.}
\label{algo:degree announcement}
$V':=V$\\
\For{$\lceil \log n \rceil$ iterations}{
  send $\delta'_i:=|{\cal N}_i\cap V'|$ to all nodes\\
  $V':=V'\setminus\{j\in V'\,|\,\delta_j'\leq 4A\}$
}
\end{algorithm}

The aim of this ``announcement phase'' is that for all iterations, the nodes
will know in advance which nodes are of high degree, which are of low degree,
and which nodes are the delegates of which other nodes. As all this information
can be inferred from the degree distributions at the beginning of each
iteration, which by itself is also a function of the degrees in the previous
iteration, the above routine performs this task.

Our goal is now to show that we can ``merge'' the further communication of all
iterations such that the total running time is bounded by $\BO(\lceil A^2/n
\rceil)$. Note that nodes satisfy up to three roles during the execution of the
algorithm: they may act as (i) high-degree nodes, (ii) delegates, and (iii) low-degree nodes.
However, according to \claimref{claim:enough_labels}, during the entire
execution of the algorithm, the total number of delegates is bounded by
\begin{equation*}
\sum_{i=1}^{\infty} \frac{n}{2^{i-1}}=2n.
\end{equation*}
We conclude that we can assign delegates in a way such that each node acts as
delegate in at most two iterations. Furthermore, each node is a low-degree node
in exactly one iteration, as afterwards it is eliminated from the subgraph
induced by active nodes. Therefore, the asymptotic bounds from
Claims~\claimref{cl:delegates_exchange} and~\claimref{cl:distribute_low_degree}
can be shown analogously also for the merged execution. Regarding
\claimref{cl:distribute_high_degree} observe that since the number of active
nodes decreases exponentially, no node sends more than $2n$ messages in its role
as high-degree node during the course of the algorithm. Overall, we obtain the
same asymptotic running time bound of $\BO(\lceil A^2/n \rceil)$ for the
communication performed by all iterations of the algorithm as we did before for a single one.
Adding the initial $\BO(\log n)$ rounds for determining the active nodes in each
iteration, the claimed running time bound follows.
\end{proof}

Furthermore, we can utilize the ``excess capacity'' of the communication system
in case $A^2\ll n$ to further reduce the number of iterations.
\begin{corollary}\label{cor:base}
\emph{TriArbor} can be modified to run in $\BO(A^2/n+\log_{2+n/A^2} n)$ rounds.
\end{corollary}
\begin{proof}
Instead of choosing the threshold for low-degree nodes to be $4A$, we pick
$\max\{4A,\lceil\sqrt{n}\rceil\}$. If $4A\geq \lceil\sqrt{n}\rceil$ the
algorithm behaves as before. Otherwise, we have that in each iteration at most
\begin{equation*}
\frac{1}{\sqrt{n}}\sum_{i\in V'}\delta_i'\leq \frac{2An'}{\sqrt{n}}
\end{equation*}
remain active, implying that all nodes are eliminated in $\BO(\log_{2+n/A^2}n)$
rounds.

It remains to show that if $\lceil\sqrt{n}\rceil\geq 4A$, all iterations can be
executed in parallel in $\BO(1)$ rounds. Observe that
Claims~\ref{cl:distribute_high_degree} and~\ref{cl:delegates_exchange} hold for
any choice of the threshold. Hence, as the number of nodes decreases
exponentially also if $\lceil\sqrt{n}\rceil\geq 4A$, the distribution of
high-degree nodes' neighborhoods and the communication among delegates can be
performed in $\BO(1)$ rounds in total. Regarding the messages sent by low-degree
nodes, in total less than $\lceil\sqrt{n}\rceil^2\leq 2n$ (instead of $16A^2$)
messages need to be conveyed, and each delegate receives at most
$\lceil\sqrt{n}\rceil^2\leq 2n$ messages. As each node is delegate at most
twice, this requires $\BO(1)$ rounds as well. Hence, taking into account
\theoremref{theorem:tri_arbor} and \corollaryref{cor:parallelizing-TriArbor},
the statement follows.
\end{proof}

It remains to remove the dependence of the algorithm on knowledge on $A$.
\begin{corollary}\label{cor:tri_arbor_uniform}
A variant of \emph{TriArbor} can be executed successfully in
$\BO(A^2/n+\log_{2+n/A^2} n)$ rounds with no prior knowledge of $A$.
\end{corollary}
\begin{proof}
Denote by $\bar{\delta}':=(\sum_{i\in V'}\delta'_i)/n'$ the average degree of
the graph of currently active nodes $G'$. Instead of setting the threshold for
high-degree nodes to $\max\{4A,\lceil\sqrt{n}\rceil\}$ as in
\corollaryref{cor:base}, we pick $\max\{2\bar{\delta}',\lceil\sqrt{n}\rceil\}$.
We have that
\begin{equation*}
\frac{1}{2\bar{\delta}'}\sum_{i\in V'}\delta'_i=\frac{n'}{2},
\end{equation*}
i.e., still at least half of the active nodes are eliminated in each iteration.
Moreover,
\begin{equation*}
2\bar{\delta}'=\frac{2}{n'}\sum_{i\in V'}\delta'_i\leq 4A,
\end{equation*}
hence, arguing analogously to \corollaryref{cor:parallelizing-TriArbor}, we can
perform all iterations together in $\BO(\lceil A^2/n \rceil)$ rounds.
\end{proof}

\begin{remark}
In~\cite{chiba85} it is shown that, for any graph, $A\in \BO(\sqrt{|E|+n})$.
Plugging this bound into the running time guaranteed
by~\corollaryref{cor:tri_arbor_uniform} yields the observation that
the round complexity achieved by \emph{TriArbor} is always in $\BO(|E|/n+\log
n)$, that is up to an additive logarithm the best we have shown so far (recall
\claimref{claim:enough_labels} yields $\BO(|E|/n)$).
\end{remark}

%% file: randomized.tex
\section{Randomization}\label{sec:randomized}

Our randomized algorithm does not exhibit an as well-structured communication
pattern as the presented deterministic solutions, hence it is difficult to
efficiently organize the exchange of information by means of a deterministic
subroutine. Therefore, we make use of a randomized routine from~\cite{lenzen11}.
\begin{theorem}[\cite{lenzen11}]\label{thm:randomized_messaging}
Given a bulk of messages such that:
\begin{compactenum}
\item No node is the source of more than $n$ messages.
\item No node is the destination of more than $n$ messages.
\item Each source knows the content of its messages.
\end{compactenum}
For any predefined constant $c>0$, all messages can be delivered in $\BO(1)$
rounds with high probability (w.h.p.), i.e., with probability at least
$1-1/n^c$.
\end{theorem}
We try to give some intuition on why this theorem is true.
A key idea is that, using randomization, it is possible to first distribute a
fairly large fraction of the messages in a roughly balanced manner, i.e., such
that $n-o(n)$ messages for each destination can be delivered by each node
sending at most $\BO(1)$ messages to the respective destination. Subsequently,
we can make ``more effort'' to distribute the remaining $o(n^2)$ messages
evenly. To this end, these messages are duplicated and sent redundantly to
different randomly chosen relay nodes. The number of copies is limited in order
to not overload the network. This results in an exponentially amplified
probability to succeed in delivering each message. Hence, after one iteration of
this scheme, we will have much less messages to deliver, enabling to increase
the number of redundant copies used for each message further, and so on.
Repeated application results in delivery of all messages in $\BO(1)$ rounds.

\subsection{The algorithm}

When sampling randomly for triangles, we would like to use the available
information as efficiently as possible. To this end, on the
first iteration of \algref{algo:tri1} all nodes sample randomly chosen induced
subgraphs of a certain size and examine them for triangles. On subsequent
iterations the size of the checked subgraphs is increased. Checking a subgraph
of size $s$ requires to learn about $\BO(s^2)$ edges, while it tests for
$\Theta(s^3)$ potential triangles.
If $s\in \Theta(\sqrt{n})$, it thus takes a linear number of messages to collect
the induced subgraph at some node and test for triangles. Using the subroutine
from \theoremref{thm:randomized_messaging}, each node can sample such a graph
in parallel in $\BO(1)$ rounds. Intuitively, this means to sample
$\Theta(n^{5/2})$ subsets of three vertices in constant time. As
$|\binom{V}{3}|\in \Theta(n^3)$, one therefore can expect to find a
triangle quickly if at least $\Omega(\sqrt{n})$ triangles are present in $G$. If
less triangles are in the graph, we need to sample more. In order to do
this efficiently, it makes sense to increase $s$ instead of just reiterating
the routine with the same set size: The time complexity grows quadratically,
whereas the number of sampled $3$-vertex-subsets grows cubically. Finally, once
the running time of an iteration hits $n^{1/3}$, we will switch to deterministic
searching to guarantee termination within $\BO(n^{1/3})$ rounds. Interestingly,
the set size of $s=n^{2/3}$ corresponding to this running time ensures that even
a single triangle in the graph is found with constant probability.

\begin{algorithm}[t!]
\caption{TriSample at node $i$.}\label{algo:tri1}
$s:=\sqrt{n}$
\While{$s<n^{1/3}$}{
  choose a uniformly random subset of $s$ nodes $C_i$\\
  \For{$j\in C_i$}{
    send the member list of $C_i$ to $j$
  }
  \For{received member list $C_j$ from $j$}{
    send ${\cal N}_i\cap C_j$ to $j$
  }
  $E_i:=\emptyset$\\
  \For{received ${\cal N}_j\cap C_i$ from $j$}{
    \For{$k\in {\cal N}_j\cap C_i$}{
      $E_i:=E_i\cup \{j,k\}$
    }
  }
  \If{$G_i:=(V,E_i)$ contains a triangle}{
    send ``triangle found'' to all nodes
  }
  \If{received ``triangle found''}{
    return \textbf{true}
  }
  \Else{
    $s:=2s$
  } 
}
run TriPartition and return its output // switch to deterministic strategy
\end{algorithm}

\subsection{Round complexity}

Our first observation is that the last iteration dominates the round
complexity of the algorithm.
\begin{lemma}\label{lemma:tri_sample_complexity}
If \emph{TriSample} terminates after $m$ iterations, the round complexity is
in $\BO(2^{2m})$ with high probability.
\end{lemma}
\begin{proof}
Let $s_k$ denote $s$ in the $k^{th}$ iteration, hence
$s_1=\sqrt{n},s_2=2\sqrt{n},...,s_m=2^{m-1}\sqrt{n}$. Clearly, in the
$k^{th}$ iteration, every node $i$ sends out exactly $s_k^2$ messages to nodes
$j$ informing them about $C_i$. Since the sets $C_{i}$ are chosen
independently, by Chernoff's bound with high probability every node $j$ in the
$k^{th}$ iteration is a member of $\BO(s_k)$ sets $C_i$, and therefore
receives $\BO(s_k)$ such subsets. It follows that, w.h.p., it will respond with
in total at most $\BO(s_k^2)$ messages telling the respective nodes $i$ about
${\cal N}_j\cap C_i$. The recipients of this messages will have to bear a load
of at most $|C_i|^2=s_k^2$. By \theoremref{thm:randomized_messaging},
these message exchanges may be accomplished in $\BO(\lceil s_k^2/n \rceil)$
rounds w.h.p. If the algorithm terminates after the $m^{th}$ iteration, the
overall round complexity is therefore in $\BO(\lceil \sum_{k=1}^m
s_k^2/n \rceil)=\BO(\lceil s_m^2/n \rceil)$.
\end{proof}
\begin{corollary}
\label{cor:Round-complexity-to-achieve-thresholds}If \emph{TriSample} is
guaranteed to find a triangle with probability $1-\varepsilon/2$ once $s$ passes
some threshold $s(\varepsilon)$, then the round complexity to find a triangle
with probability $1-\varepsilon$ is $\BO(\lceil s(\varepsilon)^2/n \rceil)$
w.h.p.
\end{corollary}
\begin{proof}
By \lemmaref{lemma:tri_sample_complexity} and the union bound.
\end{proof}
\begin{remark}
\label{rem: TriSample is in O(n^1/3)}Note that $s_m\leq n^{2/3}$
by the loop condition and afterwards the algorithm simply executes
\emph{TriPartition}. The round complexity is therefore in $\BO(n^{1/3})$ with
high probability.
\end{remark}

\subsubsection{Proof overview}

Our aim is to bound the number of iterations needed to detect a triangle with
probability at least $1-\varepsilon$, as a function of the number of triangles in
the graph. Let $T\subset\binom{V}{3}$ denote the set of triangles in $G$, where
$\left|T\right|=t$. 

On an intuitive level, the triangles are either scattered (i.e., rarely
share edges) or clustered. If the triangles are scattered, then applying
the inclusion-exclusion principle of the second order will give us
a sufficiently strong bound on the probability of success. If the
triangles are clustered, then there exists an edge that many of them
share. Finding that specific edge is more likely than finding any
specific triangle, and given this edge is found, the probability to
find at least one of the triangles it participates in is large.

\subsubsection{Bounding the probability of success using the inclusion-exclusion
principle}

We know by the inclusion-exclusion principle that
\begin{equation*}
Pr[\text{a triangle is found}] \geq t\cdot Pr[\text{exactly one triangle is
found}] -\sum_{a\neq b\in T}Pr[\text{at least \ensuremath{a} and
\ensuremath{b} are found}].
\end{equation*}
For every $a\neq b\in T$ there are three cases to consider:
\begin{compactenum}
\item $a$ and $b$ are disjoint, that is $a\cap b=\emptyset$.
\item $a$ and $b$ share a single vertex, $\left|a\cap b\right|=1$.
\item $a$ and $b$ share an edge, $\left|a\cap b\right|=2$.
\end{compactenum}
Observe that for every constant $r$ and set of vertices $V_{0}$
s.t.\ $\left|V_{0}\right|=r$, it holds that:
\begin{equation}
(s_m/(n-s_m+r))^{r}\leq Pr\left[V_{0}\text{ is chosen in the
\ensuremath{m^{th}} iteration}\right]\leq(s_m/(n-s_m))^{r}.\label{eq:Bound
on probability to choose a subset of vertices}
\end{equation}
\begin{definition}
$T_r\in\binom{T}{2}$ is the set of pairs of distinct triangles
in $G$ that have together exactly $r$ vertices. Denoting $t_r=|T_r|$,
clearly $t_4+t_5+t_6=\binom{t}{2}=|\binom{T}{2}|$.
\end{definition}
Define
\begin{equation*}
P_m=Pr[\text{triangle found in iteration \ensuremath{m}}]
\end{equation*}
and 
\begin{equation*}
p_m=Pr[\text{node \ensuremath{i} found triangle in iteration \ensuremath{m}}].
\end{equation*}
For symmetry reasons the latter probability is the same for each node $i$.
\begin{claim}
\label{cl: p_m implies P_m}For $0<\varepsilon<2$, if
$p_m\geq \ln(2/\varepsilon)/n$ then $P_m\geq 1-\varepsilon/2$.
\end{claim}
\begin{proof}
Recall that each node $i$ chooses $C_i$ independently. Consequently, the
probability of no triangle being found in the $m^{th}$ iteration is at most
\begin{equation*}
\left(1-\frac{\ln(2/\varepsilon)}{n}\right)^n
= \left(\left(1-\frac{1}{n\ln^{-1}(2/\varepsilon)}\right)^{n
\ln^{-1}(2/\varepsilon)}\right)^{\ln(2/\varepsilon)}
\leq e^{-\ln(2/\varepsilon)}
= \frac{\varepsilon}{2}.
\end{equation*}
\end{proof}
With the above notations, we combine \equalityref{eq:Bound on probability to
choose a subset of vertices} with the inclusion-exclusion principle to obtain:
\begin{equation}
p_m\geq t\cdot\left(\frac{s_m}{n-s_m+3}\right)^3
-\sum_{k=4}^6 t_k\cdot\left(\frac{s_m}{n-s_m}\right)^k.
\label{eq:inc. exc. bound}
\end{equation}

Recall that we distinguish between the cases of ``scattered'' and ``clustered''
triangles. We now give these expressions a formal meaning by defining a
threshold for $t_{4}$ in terms of $t$ and a critical value $s(\varepsilon)$ of
$s_m$ that is
$s(\varepsilon):=\max\{2n^{2/3}t^{-1/3}\ln^{1/3}(2/\varepsilon),2\sqrt{n\ln
(2/\varepsilon)}\}$. The critical value stems from either of the following cases:
\begin{compactenum}
\item Scattered triangles - we wish to sample as many triangles as possible, and
the number of triangles sampled grows cubically in $s_m$. The $n^{2/3}$ factor
in the numerator reflects the fact that $s_m=n^{{2/3}}$ would imply that each
triangle is sampled with constant probability.\footnote{Observe that
\emph{TriPartition} samples exactly $n^{{2/3}}$ vertices per node in a way covering all
subsets of 3 nodes.} Clearly having a lot of triangles in general improves the
probability of success, hence the division by $t^{-1/3}$.
\item Clustered triangles - it may be the case that all triangles share a single
edge, hence we must sample this edge with probability at least $1-\varepsilon/2$.
For $s_m=\sqrt{n}$ each node samples $\Theta(n)$ edges, hence each edge is
sampled with constant probability.
\end{compactenum}

\subsubsection{Scattered triangles}

Assume $t_{4}\leq t n/(2 s(\varepsilon))$.
\begin{lemma}
\label{lem:scattered triangles}If $t_{4}\leq t n/(2 s(\varepsilon))$ and $n$ is
sufficiently large, then a triangle will be found with probability at least
$1-\varepsilon/2$ in any iteration where $s_m\geq s(\varepsilon)$.
\end{lemma}
\begin{proof}
We rewrite \equalityref{eq:inc. exc. bound} as
\begin{equation*}
p_{m}\geq\frac{t\cdot s_{m}^{3}(n-s_{m})^{3}-t_{4}
\cdot s_{m}^{4}\cdot(n-s_{m})^{2}-t_{5}s_{m}^{5}
\cdot(n-s_{m})-t_{6}s_{m}^{6}}{(n-s_{m})^{6}}.
\end{equation*}
Due to the loop condition in \emph{TriSample}, $s_m\leq n^{{2/3}}\in (1-o(1))n$,
therefore
\begin{equation*}
p_m \geq \frac{(1-o(1))ts_m^3 n^3-t_4 s_m^4 n^2-t_5 s_m^5 n-t_6 s_m^6}{n^6}
= \frac{s_m^3((1-o(1))tn^3-t_4 s_m n^2-t_5 s_m^2 n-t_6 s_m^3)}{n^6}.
\end{equation*}

As $t_4\leq t n/(2s(\varepsilon))\leq tn/(2 s_m)$, this can be estimated
further by
\begin{equation*}
p_m\geq\frac{s_m^3((\frac{1}{2}-o(1))t n^3-t_5 s_m^2 n-t_6 s_m^3)}{n^6}.
\end{equation*}
By definition, $t_{4}+t_{5}+t_{6}=\binom{t}{2}\leq{t^{2}}/{2}$, therefore there
exist $\beta,\gamma\geq0$ such that $t_{5}=\beta t{}^{2},t_{6}=\gamma t{}^{2}$
and $\beta+\gamma\leq1/2$. Using this notation,
\begin{equation*}
p_m\geq\frac{s_m^3((\frac{1}{2}-o(1))t n^3-\beta
t^2 s_m^2 n-\gamma t^2 s_m^3)}{n^6}
= \frac{s_m^3 t((\frac{1}{2}-o(1))n^3-\beta t s_m^2 n-\gamma t s_m^3)}{n^6}.
\end{equation*}

By the loop condition in \emph{TriSample}, $s_{m}\leq n^{2/3}$.
Recalling that we assume $t\in o(n^{2/3})$, this becomes
\begin{equation*}
p_m \geq \frac{s_m^3 t((\frac{1}{2}-o(1))n^3-\beta t
n^{\frac{7}{3}}-\gamma t n^2)}{n^6}
\geq \frac{s_m^3 t(\frac{1}{2}-o(1))n^3}{n^6}.
\end{equation*}
Given that
$s_m\geq s(\varepsilon)\geq 2n^{2/3}\ln^{1/3}(2/\varepsilon)/t^{1/3}$ and, we
have for sufficiently large $n$ that
\begin{equation*}
p_m \geq\frac{s_m^3 t(\frac{1}{2}-o(1))n^3}{n^6}
=\frac{(\frac{1}{2}-o(1))s_m^3 t}{n^3}
\geq\frac{2t n^2\ln(2/\varepsilon)}{t n^3}
=\frac{2\ln(2/\varepsilon)}{n}.
\end{equation*}
By \claimref{cl: p_m implies P_m} this implies that the probability of finding a
triangle in iteration $m$ is at least $1-\varepsilon/2$.
\end{proof}

\subsubsection{Clustered triangles}

Assume $t_{4}>{t\cdot n}(2\cdot s_{m})$. The strategy employed
here is to show that due to the bound on $t_{4}$, there exists an
edge shared by many triangles. Subsequently the analysis focuses on
this edge, showing that the probability to sample this edge and find
a triangle containing it is sufficiently large.
\begin{definition}
For each edge $e\in E$, define $\Delta_{e}=\left|\left\{
T_{i}:e\subseteq T_{i}\right\} \right|$. In other words, $\Delta_{e}$ is the
number of triangles that $e$ participates in. Denote $\Delta_{\max}=\max_{e\in
E}\Delta_{e}$.
\end{definition}
\begin{lemma}
\textup{\label{lem:.DeltaMax > 2t_4 /
3t}$\Delta_{\max}\geq2t_{4}/{3t}$.}\end{lemma}
\begin{proof}
Consider a figure consisting of two triangles sharing an edge (this
is basically $K_{4}$ with one edge removed). We count the occurrences
of this figure in $G$ in two different ways:
\begin{compactenum}
\item Observe that $t_{4}$ counts just that.
\item Pick one of the $t$ triangles from $T$, choose one of its 3 edges,
denote it $e$. Choose one of the other $\Delta_{e}-1$ triangles
that share $e$ to complete the figure. Note that this counts every
figure exactly twice, since we may pick either of the two triangles in
the figure to be the first one. By definition $\Delta_{e}-1\leq\Delta_{\max}-1$,
hence we count at most $3t(\Delta_{\max}-1)/2$ occurrences.
\end{compactenum}
By comparing 1.\ and 2.\ we conclude that $t_4\leq 3t(\Delta_{\max}-1)/2$,
completing the proof.
\end{proof}
\begin{remark}
The tightness of this bound can be confirmed by examining a complete graph.
\end{remark}
\begin{lemma}
\label{lem:clustered triangles}If $t_4>t n/(2\cdot s(\varepsilon))$
then a triangle will be found with probability at least $1-\varepsilon/2$
in any iteration where $s_m\geq s(\varepsilon)$.
\end{lemma}
\begin{proof}
Assume WLOG that $e_{\max}=\{x,y\}$ is an edge shared by $\Delta_{\max}$
triangles. The probability of a node choosing both endpoints of $e_{\max}$ is
$s_m(s_m-1)/(n(n-1))\geq 0.99
s_m^2/n^2$ (for large values of $n$, as $s_m\geq\sqrt{n}$). Given that
this edge is chosen, the probability of missing all of the $\Delta_{\max}$
vertices that complete a triangle with $e_{\max}$ is at most
$(1-\Delta_{\max}/n)^{s_{m}-2}$. By \lemmaref{lem:.DeltaMax > 2t_4 / 3t} and our
assumption on $t_{4}$, we deduce that $\Delta_{\max}\geq n/(3s_{m})$, therefore the
probability of a specific node missing all triangles comprising $e_{\max}$,
conditional to $e_{\max}$ being chosen, is at most
$(1-1/(3s_m))^{s_m-2}\leq e^{-1/3}/0.99$ (for large values of $n$). Fixing some
node $i$, we obtain that
\begin{eqnarray*}
p_m&\geq & Pr\left[i \text{ finds a triangle with }
e_{\max}|x,y\in C_i\right]\cdot Pr[x,y\in C_i]\\
&\geq & \frac{(0.99-e^{-1/3})s_m^2}{n^2}\\
&\geq & \frac{s(\varepsilon)^2}{4n^2}\\
&\geq & \frac{\ln(2/\varepsilon)}{n}.
\end{eqnarray*}

Applying \claimref{cl: p_m implies P_m}, we conclude $P_{m}\geq 1-\varepsilon/2$.
\end{proof}

\subsubsection{Deriving the Bound on the Round Complexity}
\begin{definition}
$m(n,t,\varepsilon)$ is the minimal integer such that $s_{m(n,t,\varepsilon)}\geq
s(\varepsilon)$.
\end{definition}
\begin{corollary}
\label{cor:Critical-value-for-m}Given that $G$ contains at least $t$ triangles,
for every $\varepsilon>0$, with probability at least $1-\varepsilon/2$
\emph{TriSample} terminates at the latest in iteration $m(n,t,\varepsilon)$.
\end{corollary}
\begin{proof}
Combine Lemmas \ref{lem:scattered triangles} and \ref{lem:clustered triangles}.
\end{proof}
\begin{theorem}
\label{thm: TriSample Theorem}Given that $G$ contains
at least $t$ triangles, for every $\varepsilon\geq 1/n^{\BO(1)}$,
with probability at least $1-\varepsilon$ \emph{TriSample} terminates within
$\BO(\min\{n^{1/3}t^{-2/3}\ln^{2/3}\varepsilon^{-1}
+\ln\varepsilon^{-1},n^{1/3}\})$ rounds. It always outputs the correct result.
\end{theorem}
\begin{proof}
By \corollaryref{cor:Critical-value-for-m}, the algorithm terminates
with probability $1-\varepsilon/2$ after no more than $m(n,t,\varepsilon)$ iterations.
By \corollaryref{cor:Round-complexity-to-achieve-thresholds}, it thus terminates
with probability $1-\varepsilon$ within $\BO(2^{2m(n,t,\varepsilon)})$ rounds.
By Remark \ref{rem: TriSample is in O(n^1/3)} the round complexity is always in
$\BO(n^{{1/3}})$ with high probability, altogether showing the stated bound.

Correctness follows from the fact that the algorithm terminates if
it either finds a triangle, or after executing \emph{TriPartition}, according
to~\theoremref{thm:TriPartition-determines-correctly in n^1/3} with the correct output.
\end{proof}
\begin{corollary}
Algorithm \emph{TriSample} terminates within $\BO(\lceil
n^{1/3}/(t+1)^{2/3}\rceil)$ rounds in expectation and within $\BO(\max\{
n^{1/3}\ln^{2/3}n/t^{2/3} + \ln n,n^{1/3}\})$ rounds w.h.p.
\end{corollary}
\begin{remark}
We can make sure the algorithm \emph{always} terminates within $\BO(n^{{1/3}})$
rounds by stopping it after $n^{{1/3}}$ rounds and switching
to \emph{TriPartition} even if $s_{m}<n^{{2/3}}$.\end{remark}
\begin{corollary}
For every $\varepsilon\geq 1/n^{\BO(1)}$, it is possible to distinguish with
probability $1-\varepsilon$ between the cases that $G$ is triangle-free
and that $G$ has at least $t_0\geq 1$ triangles within
$\BO(t_{0}^{-2/3}n^{1/3}\ln^{2/3}(1/\varepsilon)+\ln(1/\varepsilon))$
rounds.
\end{corollary}
\begin{proof}
Set $s\leq 2\max\{t_{0}^{-1/3}2n^{2/3}\ln^{1/3}(1/\varepsilon),
2\sqrt{n\ln(1/\varepsilon)}\}$ to be the loop condition \emph{TriSample}. If no
triangle has been found during the loop, we output that $G$ is triangle-free.
The running time bound follows from
\corollaryref{cor:Round-complexity-to-achieve-thresholds}, and correctness with
probability $1-\varepsilon$ is due to \theoremref{thm: TriSample Theorem}.
\end{proof}

\subsection{Tightness of the Analysis}
\begin{claim}
The running time bound from \theoremref{thm: TriSample Theorem} is
asymptotically tight, that is, there are graphs for which \emph{TriSample} runs
with probability at least $\varepsilon$ for
$\Omega(n^{1/3}t^{-2/3}\ln^{2/3}\varepsilon^{-1})$ or
$\Omega(\ln\varepsilon^{-1})$ rounds, respectively.
\end{claim}
\begin{proof}
Consider a graph $G$ with $t<n-2$ triangles, all sharing a specific edge
$e_{0}$. To find a triangle, some node must sample both ends of $e_{0}$, and
this happens with probability ${s_{m}(s_{m}-1)}/(n(n-1))$ per node. The
probability that all nodes miss $e_{0}$ is at least
$(1-s_{m}(s_{m}-1)/(n(n-1)))^{n}$. If $s_m\in o(\sqrt{n\ln(1/\varepsilon)})$
then this probability is in $1-\omega(\varepsilon)$.

Consider a graph $G$ with $t$ disjoint triangles $t<n/3$. The probability of a
specific node to miss a specific triangle is at least
$1-(s_m(s_m-1)(s_m-2))/(n(n-1)(n-2))\geq 1-((s_m-2)/n)^3$. By the
union bound, the probability of a specific node missing all triangles is at
least $1-t((s_m-2)/n)^3$. The probability that all nodes miss all
triangles is therefore at least $(1-t((s_m-2)/n)^3)^n$. Assuming that $s_m\in
o(t^{-1/3}n^{2/3}\ln^{1/3}(1/\varepsilon))$, this is in
$(1-o(n^{-1}\ln(1/\varepsilon))^n\subseteq\omega(\varepsilon)$.
\end{proof}

%% file: ack.tex
\subsection*{Acknowledgements}
The authors would like to thank Shiri Chechik for suggesting
\algref{algo:round_robin}, and Brendan McKay for his proof of
\lemmaref{lem:.DeltaMax > 2t_4 / 3t}. Danny Dolev is incumbent of the Berthold
Badler Chair in Computer Science. Christoph Lenzen has been supported by the
Swiss National Science Foundation and the Society of Swiss Friends of the
Weizmann Institute of Science. This research project was supported in part by
The Israeli Centers of Research Excellence (I-CORE) program, (Center No. 4/11),
by the Google Inter-university
center for Electronic Markets and Auctions, and
by the Kabarnit Consortium, administered by the office of the
Chief Scientist of the Israeli ministry of Industry and Trade and Labor.